\documentclass[12pt]{article}
\usepackage{graphicx}
\usepackage{amsfonts,amsmath}
\usepackage[mathscr]{eucal}
\usepackage{amssymb}
\usepackage{amsthm}
\usepackage{bbold}
\theoremstyle{plain}

\newtheorem{defi}{Definition}
\newtheorem{prop}{Proposition}

\newcommand{\be}{\begin{eqnarray}}
\newcommand{\ee}{\end{eqnarray}}
\newcommand{\bc}{\begin{center}}
\newcommand{\ec}{\end{center}}
\newcommand{\nn}{\nonumber \\}
\newcommand{\lb}{\label}
\newcommand{\p}[1]{(\ref{#1})}

\begin{document}

\begin{titlepage}

\vspace*{0.2cm}

\renewcommand{\thefootnote}{\star}
\begin{center}

{\LARGE\bf  Group manifolds and homogeneous spaces with HKT geometry: the role of automorphisms.}

\vspace{2cm}

{\Large A.V. Smilga} \\

\vspace{0.5cm}

{\it SUBATECH, Universit\'e de
Nantes,  4 rue Alfred Kastler, BP 20722, Nantes  44307, France. }

\end{center}
\vspace{0.2cm} \vskip 0.6truecm \nopagebreak

   \begin{abstract}
\noindent We present a new simple proof of the fact that certain group manifolds  as well as certain homogeneous spaces $G/H$ of dimension $4n$ admit a triple of integrable complex structures that satisfy the quaternionic algebra and are covariantly constant with respect to the same torsionful Bismut connection, i.e. exhibit the HKT geometry. The key observation is that different complex structures are interrelated by  automorphisms of the Lie algebra. To construct the quaternion triples, one only needs to construct the proper automorphisms, which is a more simple problem. 
   \end{abstract}

\end{titlepage}

\setcounter{footnote}{0}

\setcounter{equation}0
\section{Motivation}
HKT manifolds,\footnote{``HKT" means ``Hyper-K\"ahler with Torsion". This name is somewhat misleading, because these manifolds are not hyper-K\"ahler and not even K\"ahler, but it is well established in the literature, no better term has been proposed, and we will use it.}
discovered first by physicists \cite{HKT}, attracted also a considerable interest of mathematicians \cite{HKT-math}. Mathematicians got interested  in the HKT manifolds because of  their rich and nontrivial geometric structure. The physical interest stems from the fact that supersymmetric sigma models defined on HKT target spaces enjoy extended supersymmetries, which makes them relevant for some field theory applications.

I said ``field theory", and there are, indeed, interesting (1+1)-dimensional supersymmetric sigma models that live on the HKT manifolds \cite{Spindel}, but there are also more simple {\it mechanical} models where the dynamic variables---the coordinates of the manifold  and their Grassmann superpartners---depend only on time, and there is no spatial dependence.

Consider, in particular, the sigma model with the following superfield action \cite{Coles},
\be
\lb{SgenN1}
S \ =\ \frac i2 \int  d\theta dt \, g_{MN}({\cal X}) {\dot {\cal X}}^M {\cal D} {\cal X}^N -  
\frac 1{12}
\int d\theta   dt \, C_{KLM} {\cal D} {\cal X}^K {\cal D} {\cal X}^L {\cal D} {\cal X}^M \, ,
    \ee
    where ${\cal X}^M(t,\theta) \  = \ x^M(t) + i \theta \Psi^P(t)$ are $D=4n$ ${\cal N} = 1$
    superfields; $x^M$ are the coordinates on the manifold (on its one particular chart), 
    $\Psi^M$ are their Grassmann superpartners, $g_{MN}$ is the metric, $C_{KLM}$ is the totally antisymmetric torsion tensor and ${\cal D} = \partial/\partial \theta - i\theta \partial/\partial t$ is the supersymmetric covariant derivative.
    
    The ${\cal N} = 1$ supersymmetry of the action \p{SgenN1}  is manifest.
    
    \begin{prop}
    If the action \p{SgenN1} stays invariant under three extra supersymmetries
    \be
    \delta  {\cal X}^M  \ =\ \epsilon_p (I_p)_N{}^M {\cal D} {\cal X}^N \,,
     \ee
     where $\epsilon_{p=1,2,3}$ are  Grassmann transformation parameters and $I_p \equiv 
     (I,J,K)$ is a triple of integrable complex structures satisfying the quaternionic algebra
      \be
	      \lb{quatern-alg}
      I_p I_q \ =\ -\delta_{pq} + \varepsilon_{pqs} I_s \,,
       \ee
    the target space of this model is an HKT manifold: the complex structures $I_p$ are covariantly constant with respect to the Bismut connection,
    \be 
    \lb{Bismut}
    G^M_{NP} \ = \ \Gamma^M_{NP} + \frac 12 g^{MK} C_{KNP}\,,
     \ee 
where $\Gamma^M_{NP}$ are the ordinary symmetric Christoffel symbols and $C_{KNP}$ is the totally antisymmetric torsion tensor.
 \end{prop}
 
 This assertion was proved in the component language in \cite{HKT} and, in the language of ${\cal N} = 1$ superfields, in recent \cite{FIS-nonlin}.
 
 A large class of HKT manifolds are group manifolds.\footnote{Note, however that not {\it all} HKT manifolds have group nature. For example, the Delduc-Valent HKT manifold \cite{DV}
 is not associated with any group.} 
 The full list of the HKT group manifolds for the compact groups including only one non-Abelian factor, is given  below. 
 \be
\lb{list-group-HKT}
&&SU(2l+1), \ SU(2l) \times U(1), \ Sp(l) \times [U(1)]^l , \nn
&&SO(2l+1) \times [U(1)]^l, 
SO(4l) \times [U(1)]^{2l}, \ SO(4l+2) \times [U(1)]^{2l-1} ,   \nn 
&&G_2 \times [U(1)]^2 , \  F_4 \times   [U(1)]^4  , \ E_6 \times [U(1)]^2  , 
\ E_7 \times [U(1)]^7, \nn && E_8 \times [U(1)]^8\, .
 \ee
 The fact that all these manifolds are HKT was proven first in \cite{Spindel}. Another proof was suggested in \cite{Joyce}. In this paper we present still another proof, which includes some salient features of the proofs in \cite{Spindel,Joyce}, but is more simple and explicit.  Our explicit constructions are based on the observation that the complex structures $(I,J,K)$ satisfying the quaternionic algebra \p{quatern-alg} can be chosen so that $J$ and $K$ are derived from $I$ by certain automorphisms of the relevant Lie algebra ${\mathfrak g}$.
 
 Besides group HKT manifolds, there are also many HKT manifolds representing homogeneous spaces, like $SU(4)/SU(2)$ \cite{Joyce,OP}. We will discuss them in Sect. 5 of the paper.
 
\section{Basic facts and definitions}
\setcounter{equation}0
\begin{defi}
\label{defi-almost}
        A  {\sl complex manifold} is a   manifold equipped with a 
tensor field $I_{MN}$ satisfying the properties
       {\bf (i)} $ I_{MN} = -I_{NM}$; {\bf (ii)} $I_M{}^N I_N{}^P = -\delta_M^P$;  
       \be
       \lb{Nijen}
 {\rm \bf (iii)} \qquad \qquad      {\cal N}_{MN}{}^K \ =\ \partial_{[M} I_{N]}{}^K -  
I_M{}^P  I_N{}^Q \partial_{[P}  I_{Q]}{}^K  \ =\ 0 \,. 
       \ee
    The tensor $I_M{}^N$ is called the {\em
       complex structure}.
      \end{defi}
         The first two conditions in this definition imply that the manifold is even-dimensional. The third condition  is the requirement for the so-called {\sl Nijenhuis tensor} to vanish. According to the {\sl Newlander-Nirenberg theorem} \cite{NN} (see also \cite{NN-ja}), this is necessary and sufficient for {\sl integrability}, i.e.  for a possibility to introduce the  complex coordinates associated with the complex structure $I$.  
          
          Mathematicians often consider real and complex manifolds that are not equipped with a metric. But the group manifolds, which we discuss in this paper, {\it are}. Thus,  we will always assume that the metric (allowing to lift and to lower the indices) is defined. For a complex manifold (and we will show soon that an even-dimensional  group manifold is  complex), the metric expressed in complex coordinates $z^j$ has a Hermitian form:
 \be
 ds^2 \ =\ h_{j\bar k} dz^j d\bar z^{\bar k}\,, \qquad  {\rm with} \quad   \overline{h_{j\bar k}} = h_{k\bar j}\,.
 \ee        
  
  The tensor $I_M{}^M$ can be represented as 
  \be
  \lb{IMN}
  I_M{}^N \ =\ e_{MA} I_{AB} e^N_B \,,
   \ee
   where $e_{MA}$ and $e^N_B$ are the {\sl vielbeins} such that $g_{MN} = e_{MA}   e_{NA}$. $I_{AB}$ is a matrix acting in the tangent space of our group manifold, i.e. on the corresponding Lie algebra. The properties $I_{AB} = - I_{BA}$ and 
   $I^2 = -\mathbb{1}$ hold.  $I_{AB}$ is the same at all points of the manifold.         
        
\begin{defi}
A {\sl K\"ahler manifold} is a complex manifold with covariantly constant complex structure, $\nabla_M^{L.C.} I_{NQ} = 0$,
where $\nabla_M^{L.C.}$ is the ordinary Levi-Civita covariant derivative involving the Christoffel symbols.
 \end{defi}
 
 For a generic complex manifold, the complex structure is not covariantly constant in this sense, but it is covariantly constant with respect to an infinite set of torsionful connections. In particular,
 
 \begin{prop}
 The complex structure of a generic complex manifold is covariantly constant with respect to the {\sl Bismut connection} involving the structure \p{Bismut}. The totally antisymmetric torsion tensor is expressed as\footnote{The Bismut connection was introduced in  \cite{Bismut}. The beautiful formula \p{tors-Hull} was written in \cite{OP,Hull}.}
 \be
\lb{tors-Hull}
C_{MNP} \ =\ I_M{}^Q I_N{}^S I_P{}^R \, (\nabla_Q^{L.C.} I_{SR} +  \nabla_S^{L.C.} I_{RQ} +  \nabla^{L.C.}_R I_{QS} ) \,. 
 \ee
\end{prop}

\begin{defi} A {\sl hyper-K\"ahler manifold} is a manifold admitting a triple of  covariantly constant complex structures,  $\nabla_M^{L.C.} I^{p=1,2,3}_{NQ} = 0$, that satisfy the quaternion algebra \p{quatern-alg}.
\end{defi}

\begin{defi}
An HKT manifold is a manifold admitting three integrable complex structures that satisfy the algebra \p{quatern-alg} and are covariantly constant with respect to the same universal torsionful Bismut connection.
\end{defi}

One can easily prove that the dimension of both hyper-K\"ahler and  HKT manifolds is a multiple integer of 4.

\vspace{1mm}

For the group theory considerations, we use the ``physical" convention where the generators $t_A$ are Hermitian. We use the Cartan-Weyl basis where the set of all $t_A$ is subdivided into
 the orthonormal basis $t_{a = 1,\ldots, r}$ in  the Cartan subalgebra  (CSA)  $H$, the positive {\it root vectors} $E_{\alpha_j}$ and the negative root vectors $E_{-\alpha_j}$ that are Hermitially conjugate to $E_{\alpha_j}$. Then for any $h \in H$, the commutation relations
  \be
  [h, E_{\pm \alpha_j}] = \pm \alpha_j(h) E_{\pm \alpha_j}
   \ee
   hold. Here $\alpha_j (h)$ are linear forms on $H$ called the  {\it roots}.  We define also the {\it coroots} $\alpha_j^\vee$ as the elements of $H$ satisfying $\alpha_j(\alpha_k^\vee) = 2\delta_{jk}$. The coroots also satisfy the following noteworthy property:
   \be 
 \lb{coroot-norm}
\omega  =  e^{2i\pi \alpha_j^\vee} \ = \ \mathbb{1}, \ \ {\rm but} \ \ \qquad e^{i\phi \alpha_j^\vee} \ \neq \ \mathbb{1} \ \ {\rm if} \ \ \phi < 2\pi .
  \ee
    For any positive root $\alpha$ the  
 commutator $[E_\alpha, E_{-\alpha}]$ is proportional to $\alpha^\vee$.  We choose the {\it Chevalley normalization} of the root vectors where 
  \be
     \lb{E-norm}
[E_\alpha, E_{-\alpha}] \ =\ \alpha^\vee \,. 
 \ee  
  For any two positive or negative roots $\alpha, \beta$ with $\alpha + \beta \neq 0$, the commutator $[E_\alpha, E_\beta]$ is proportional to $E_{\alpha+\beta}$ if the root $\alpha+\beta$ exists. Otherwise, this commutator vanishes. In the Chevalley normalization, the following nice relation holds \cite{Bourbaki}:
   \be
   \lb{Bourbaki}
  [E_{\alpha}, E_{\beta}] \ =\ \pm (q+1) E_{\alpha + \beta} \, ,
   \ee
   where $q$ is the greatest positive integer such that $\alpha - q\beta$ is a root.

\section{Samelson theorem}
\setcounter{equation}0
We choose a matrix representation of the group and parametrize the group elements as
 \be
 \omega  \ =\ \exp\{it_M x^M\} \,,
  \ee
  where $t_M$ are the generators satisfying\footnote{C = 1/2 in the defining representation of $SU(N)$.} Tr$\{t_M t_N\} = C \delta_{MN}$  and $x^M$ are the coordinates of the group manifold. We endow the manifold with the Killing metric
   \be
   \lb{Killing}
   g_{MN} \ =\  \frac 1C {\rm Tr} \{  \partial_M \omega \, \partial_N \omega^{-1} \} \,.
    \ee 
    Then $ds^2= g_{MN} dx^M dx^N$ is invariant under multiplication of $\omega$ by any group element on the left or on the right. In the vicinity of unity, $x^M \ll 1$, the metric \p{Killing} acquires the form
     \be
     \lb{g-okrest}
     g_{MN} \ =\ \delta_{MN} - \frac 1{12} f_{MPQ} f_{NPR}\, x^Q x^R + o(x^2)\,,
      \ee
      where $f_{MPQ}$ are the structure constants. 
      The Samelson theorem says:
      \begin{prop} \cite{Samelson}
 Any even-dimensional group manifold is complex.
  \end{prop}
  
  \begin{proof}
   To prove the theorem, we should define an almost complex structure $I_{MN}$ and show that the Nijenhuis tensor for this structure vanishes. We take care in this definition that the components of the tensor $I_{MN}$ in the different points of the manifold are related to each other by the coordinate transformations generated by, say, a right group multiplication $\omega \to \omega V$. For the close points, this gives
    \be
   \lb{IMN-close}  
   I_{MN}(x) \ =\ I_{MN}(0) + \frac 12 I_M{}^Q(0) f_{NQP}\,x^P -  \frac 12 I_N{}^Q(0) f_{MQP} \, x^P  + o(x) \,.
    \ee
    The relation \p{IMN-close} can be alternatively written as
    $ I_{MN} \ =\ e_{MA} e_{NB} I_{AB}$ \,, 
     where 
     \be
     \lb{vielbein-group}
     e_{MA} \ =\ \delta_{MA} + \frac 12 f_{MAP}\, x^P - \frac 16 f_{AMR} f_{ANQ} \, x^R x^Q  \ + o(x^2)
     \ee 
     are the vielbeins satisfying $e_{MA} e_{NA} = g_{MN}$ and $I_{AB} \equiv I_{MN}(0)$ is the tangent space projection of the complex structure, the same at all the points.
     
       The complex structure \p{IMN-close} is covariantly constant with respect to the Bismut connection with the torsion tensor 
    \be
    \lb{C=f}
    C_{MNP} = f_{MNP}\,.
    \ee
     Indeed,
    it is straightforward to see that at the origin $x=0$,
    \be
 \lb{Bis-I-unity}
 \nabla_P^{(B)} I_{MN} =  \partial_P I_{MN} - \frac 12 f_{QPM} I^Q{}_N - \frac 12 f_{QPN} I_M{}^Q  =  0\,,
  \ee
  where we neglected the contribution of the ordinary Christoffel symbols $\Gamma^Q_{PM}$, which are of order $O(x)$.  Note that the torsion tensor \p{C=f} is invariant under group rotations, like the metric is, and does not depend on $x$.
       
        Let us now substitute the complex structure \p{IMN} with the vielbeins \p{vielbein-group} in the integrability condition \p{Nijen} and consider there only leading terms  $\propto O(1)$. After some simple transformations, we arrive at the identity
       \be
       \lb{f+3fI=0}
       f_{ABC} -  I_{AD} I_{BE}\, f_{DEC} -    I_{BD} I_{CE} \, f_{DEA} - 
        I_{CD} I_{AE}\, f_{DEB} \ =\ 0 \, .
         \ee
         This was derived by considering the condition \p{Nijen} at the vicinity of one particular point of the manifold with $\omega = \mathbb{1}$. But, bearing in mind the fact that $I_{AB}$ is the same everywhere and the isometry
 $G_L \times G_R$ of the metric \p{Killing}, so that any point of the group manifold can be brought to unity (where $g_{MN} = \delta_{MN}$ and $e_{AM} = \delta_{AM}$) by a group rotation, the identity \p{f+3fI=0} implies the fulfillment of the condition \p{Nijen} at all points.

  Thus, we have to find an antisymmetric matrix $I$ that squares to
   $ -\mathbb{1}$ and   satisfies  the condition \p{f+3fI=0}. The geometric problem is reduced to a pure group theory one!
  
   The tangent space of a group manifold is the Lie algebra of the corresponding group. The matrix $I_{AB}$ can be interpreted as a linear operator acting on the generators $t_A$. We define the action of $I$ on the root vectors as 
  \be
  \lb{act-I-korni}
  \hat I E_\alpha \ =\ -i E_\alpha, \qquad 
  \hat I E_{-\alpha}  \ =\ i E_{-\alpha} \,.
   \ee
        Each root vector can be represented as $E_{\pm \alpha} = t_A \pm  
   it_{A^*}$ with Hermitian $t_A$ and $t_{A^*}$. For example, the algebra $A_2 \equiv su(3)$  involves three positive root vectors:
   \be
   \lb{E-su3}
   E_\alpha &=& \left(\begin{array}{ccc} 0 & 1 & 0 \\ 0& 0&0 \\ 0&0&0 \end{array} \right) \ = \  t_1 + it_2\, , \nn
    E_\beta &=& \left(\begin{array}{ccc} 0 & 0 & 0 \\ 0& 0&1 \\ 0&0&0 \end{array} \right) \ = \  t_6 + it_7\, , \nn
     E_{\alpha+\beta} &=& \left(\begin{array}{ccc} 0 & 0 & 1 \\ 0& 0&0 \\ 0&0&0 \end{array} \right) \ = \  t_4 + it_5\,.
      \ee
     In these terms,  
   \be
   \lb{I-korni}
    \hat I t_A = t_{A^*}, \  \hat I t_{A^*} = - t_{A} \quad {\rm or\ else} \quad I_{A^* A}   \ =\  - I_{AA^*} \ =\ 1 \,.
    \ee

   If the group is even-dimensional, its Cartan subalgebra is even-dimensional.
   Choose there the orthonormal basis $t_a$, order the set $\{t_a\}$ in an arbitrary way and define  
     \be
  \lb{I-CSA}
  \hat I  t_{a_1} = t_{a_2}, \quad  \hat I  t_{a_2} = -t_{a_1}, \qquad  \hat I  t_{a_3} = t_{a_4}, \quad  \hat I  t_{a_4} = -t_{a_3}\,, \quad {\rm etc.}
    \ee
  It is clear that the matrix $I$ thus defined is antisymmetric and squares to $-\mathbb{1}$.  Let us prove that it satisfies the condition \p{f+3fI=0}.
    
    \begin{enumerate}
       
   \item Consider the L.H.S. of \p{f+3fI=0} with an arbitrary $C$ and $A,B$ associated with the same root vector: $A = A, B=A^*$. Bearing in mind that $I_{AA^*} 
I_{A^*A} = -1$, it is easy to see that the first term in \p{f+3fI=0} cancels the second one, while  the third and the fourth term vanish.

\item Let now $A$ and $B$ be associated with different root vectors $E_\alpha$ and $E_\beta$, with $t_C\equiv h$ belonging to the Cartan subalgebra. The commutator $[h, E_\alpha]$ is proportional to $E_\alpha$ and hence $f_{ABC} = 0$. In this case, all the terms in \p{f+3fI=0} vanish.

\item A somewhat less trivial case is when $A,B,C$ are associated with three different roots $\alpha,\beta,\gamma$. Note that for any such triple,  one can find a couple $(\alpha,\beta)$ such that the commutator $[E_\alpha, E_\beta]$ has no projection on $E_\gamma$. Indeed, all the nonzero commutators  of the positive root vectors have the form $[E_\alpha, E_\beta] = CE_{\alpha+\beta}$. Consider  then a couple $(\alpha, \alpha +\beta)$. The commutator  $[E_\alpha, E_{\alpha+\beta}]$ can only be equal to $E_{2\alpha + \beta}$ (if the root $2\alpha + \beta$ exists) and does not have a projection on $E_\beta$.

Let $(\alpha,\beta,\gamma)$ be such a triple. Then
$
[t_A + i t_{A^*}, t_B + i t_{B^*}]$  has no projection on  $t_C + i t_{C^*}$.
 It follows that 
\be
\lb{relat-f}
f_{ABC^*} - f_{A^*B^*C^*} \ =\  f_{A^*BC} + f_{AB^*C} \ = 0
 \ee
 (note also that the structures like $f_{ABC}$ or $f_{AB^*C^*}$ vanish). 
  Bearing  this in mind, it is easy to see that the relation \p{f+3fI=0} holds for all star attributions. For example, for $\{ABC\} \to \{ABC^*\}$, we deduce 
   \be
   \lb{chetyre-f}
   f_{ABC^*} + f_{B^*A^*C^*} - f_{CB^*A} - f_{A^*CB} \ =\ 0 \,.
    \ee
        
     \end{enumerate}

\end{proof}

\section{Quaternion triples}
\setcounter{equation}0

Our goal is to prove that certain group manifolds of dimension $4n$ are HKT, which is tantamount to say that certain groups of dimension $4n$ admit quaternion triples of the matrices $I,J,K$ that satisfy the identity  \p{f+3fI=0}.
Indeed, it follows from the consideration above that the corresponding complex
structures are integrable. In addition, the torsion tensors for each such structure are given by \p{C=f} and coincide.

The basic observation is 
\begin{prop} 
   Let $\Omega$ be an automorphism of the algebra $\mathfrak{g}$.
    Let $I$ be an antisymmetric matrix that squares to $-\mathbb{1}$ and satisfies \p{f+3fI=0}. 
    Then 
  \be
  \lb{J=OmegaI}
  J_{AB} = (\Omega I \Omega^T)_{AB}
  \ee
  has the same properties.
     \end{prop}
    \begin{proof}
   An automorphism of $\mathfrak g$ is an orthogonal matrix $\Omega_{AB}$ satisfying the condition 
    \be
    \Omega_{AD} \Omega_{BE} \Omega_{CF} f_{DEF} = f_{ABC}\,.
     \ee
     Then clearly $J^T = -J$ and $J^2 = -\mathbb{1}$. The fulfillment of \p{f+3fI=0} for the matrix \p{J=OmegaI}
    also immediately follows from the invariance of $ f_{ABC}$.
    \end{proof}
Let $I$ be a matrix \p{I-korni}, \p{I-CSA}. To prove the assertion, we have to find automorphisms that make $J$ and $K$ out of  $I$ in such a way that the triple $I,J,K$ is quaternionic. 
Before doing so in a generic case, we consider several examples.

\subsection{$SU(2) \times U(1)$}

Another name for the manifold $SU(2) \times U(1) \equiv S^3 \times S^1$ is the {\it Hopf manifold} \cite{Hopf}. We will prove that this manifold is HKT.

The $su(2)$ algebra has only one positive root vector $E_+ = t_1 +it_2$. The Cartan subalgebra of $su(2) \oplus u(1)$ includes the generator $t_3$ of $su(2)$ and $t_0$ of $u(1)$. The canonical complex structure $\EuScript{I}$ is the matrix

\be
\lb{block-I}
 {\EuScript I} \ =\ \begin{array}{c} \mbox{\scriptsize 1}\\ \mbox{\scriptsize 2}\\ \mbox{\scriptsize 3}\\ \mbox{\scriptsize 0} \end{array} \quad \left( \begin{array}{cccc} 0&-1&0&0 \\ 1&0&0&0 \\ 0&0&0&-1 \\ 0&0 & 1&0 \end{array} \right)\,. 
    \ee
    The first two lines in $\EuScript{I}$ follow from the definition \p{act-I-korni}. The third and the fourth line correspond to the sign convention $\hat {\EuScript{I}} t_3 = t_0, \hat {\EuScript{I}} t_0 = -t_3$ that we adopt.
    
Consider the automorphism 
\be
\lb{auto-su2}
\Omega: \quad t_{1,0} \to t_{1,0}, \quad t_2 \to t_3 \to -t_2 \,.
 \ee
 It brings the matrix $\EuScript{I}$ to
 \be
 \lb{block-J}
 \EuScript{J} \ =\ \Omega \EuScript{I} \Omega^T = 
 \begin{array}{c} \mbox{\scriptsize 1}\\ \mbox{\scriptsize 2}\\ \mbox{\scriptsize 3}\\ \mbox{\scriptsize 0} \end{array} \quad \left( \begin{array}{cccc} 0&0&-1&0 \\ 0&0&0&1 \\ 1&0&0&0\\ 0&-1 & 0&0 \end{array} \right)\,. 
 \ee
  The third matrix is 
  \be
  \lb{block-K}
  \EuScript{K} \ = \ \EuScript{I} \EuScript{J} \ =\ 
   \begin{array}{c} \mbox{\scriptsize 1}\\ \mbox{\scriptsize 2}\\ \mbox{\scriptsize 3}\\ \mbox{\scriptsize 0} \end{array} \quad \left( \begin{array}{cccc} 0&0&0&-1 \\ 0&0&-1&0 \\ 0&1&0&0\\ 1&0&0&0 \end{array} \right)\,.
  \ee 
  One can be easily convinced that these matrices anticommute and are quaternionic. One can further notice that all of them are self-dual,
  $\EuScript{I}_{AB} = \frac 12 \varepsilon_{ABCD} \EuScript{I}_{CD}$ with the convention
  $\varepsilon_{1230} = 1$, and the same for $\EuScript{J}$ and $\EuScript{K}$. (They would be anti-self-dual if the opposite sign convention in the last two lines of \p{block-I} were chosen.) A physicist would recognize in these matrices
  the so-called {\it `t Hooft symbols} \cite{Hooft}. 
  
  The automorphism \p{auto-su2} can be represented as $t_a \to U^\dagger t_a U$ with
  \be
  \lb{IJ-auto}
  U \ =\ \frac 1 {\sqrt{2}} \left( \begin{array}{cc} 1 & i \\ i & 1 \end{array}\right)  \ =\ \exp\left\{\frac {i\pi}2 t_1 \right\} \ =\  \exp\left\{\frac {i\pi}4 (E_+ + E_-) \right\}  \, .
   \ee
     If one chooses
     \be
  \lb{IK-auto}
  U \ =\ \frac 1 {\sqrt{2}} \left( \begin{array}{cc} 1 & 1 \\ -1 & 1 \end{array}\right)  \ =\ \exp\left\{\frac {i\pi}2 t_2 \right\} \ =\  \exp\left\{\frac {\pi}4 (E_+ - E_-) \right\}  \, ,
   \ee
   one arrives to the automorphism $\tilde{\Omega}: \quad t_{2,0} \to t_{2,0}, \ t_1 \to -t_3 \to t_1$.
   It transforms the complex structure $\EuScript{I}$ to $\tilde \Omega \EuScript{I} \tilde \Omega^T = \EuScript{K}$.

  \subsection{$SU(3)$}
  
  This is the next in complexity case. We start with constructing the complex structure $I$. The definition \p{act-I-korni} gives the matrix elements
  \be
  I_{21} = I_{54} =  I_{76} = -  I_{12} =  - I_{45} = - I_{67} \ =\ 1\,.
   \ee
   In the Cartan subalgebra, we choose the basis\footnote{This is {\it not} the standard choice adopted in the physical community. Our $t_3$ is one half of the coroot $\alpha^\vee +\beta^\vee$  rather than $\alpha^\vee/2$, as usual.} 
    \be
    \lb{basis}
   t_3 = \frac 12\,{\rm diag}(1,0,-1), \qquad t_8 = \frac 1{2\sqrt{3}}\,
  {\rm diag}(1,-2,1) \,.
   \ee
    We then define $I_{83} = - I_{38} = 1$.
  
  We see that the $8 \times 8$ matrix $I$  splits into two blocks:
  \begin{enumerate}
  \item The block in the subspace (4,5,3,8) that corresponds to
   the subalgebra $su(2) \oplus u(1)$ of $su(3)$, with $su(2)$ associated with the highest root $\alpha+\beta$.
   \item The block in the subspace (1,2,6,7) acting on the root vectors $E_{\pm \alpha}, E_{\pm \beta}$. 
  \end{enumerate}
  Each block has the form \p{block-I}: $I \ =\ {\rm diag}(\EuScript{I}, \EuScript{I} )$.
  
  To find the second complex structure, we calculate $\Omega I \Omega^T$ with the automorphism $\Omega: \ t_a \to U^\dagger t_a U$, where
  \be
  \lb{u-su3}
  U \ =\ \exp \left\{ \frac {i\pi}4 (E_{\alpha+\beta} + E_{-\alpha-\beta}) \right\} = \exp \left\{ \frac {i\pi}2 t_4 \right\} \ =\ 
  \frac 1{\sqrt{2}} \left( \begin{array}{ccc} 1& 0 &i \\0 &\sqrt{2} & 0 \\ i & 0 & 1 \end{array} \right) \,.
   \ee
   The generators from the subalgebra $su(2) \oplus u(1)$ transform in the same way as in \p{auto-su2}:
   $t_{4,8} \to t_{4,8},  \ t_5 \to t_3 \to -t_5$. To find the transformations of the root vectors 
   $E_{\pm \alpha}, E_{\pm \beta}$, we use the Hadamard formula:
   \be
   \lb{Hadamard}
   e^R X e^{-R} \ =\ X + [R,X] + \frac 12 [R, [R,X]] + \frac 16 [R,  [R, [R,X]]] + \ldots
     \ee
     with $R = -i\pi/4 (E_{\alpha+\beta} + E_{-\alpha - \beta})$. The nontrivial commutators are\footnote{Note that the commutators like $[E_{\alpha+\beta}, E_\alpha]$ vanish due to the fact that $\alpha + \beta$ is the highest root. That is why our choice  (actually, the choice of Ref. \cite{Spindel}) of the subalgebra $su(2) \oplus u(1) \subset su(3)$ is more clever than the other choices. In fact, for $SU(3)$, one could also choose it in the ordinary way in association with the root $\alpha$ rather than with $\alpha+\beta$. One could thus construct a quaternion triple of the complex structures, but this method is not universal and cannot be easily generalized to an arbitrary group, which is our goal.}
     \be
     \lb{comm-RE}
     [R, E_\alpha] \ =\ - \frac {i\pi} 4 E_{-\beta}, &\qquad& [R, E_\beta]   \ =\  \frac {i\pi} 4 E_{-\alpha}, \nonumber \\
   \, [R, E_{-\alpha}] \ = \  \frac {i\pi} 4 E_{\beta}, &\qquad& [R, E_{-\beta}]   \ =\ -\frac {i\pi} 4  E_{\alpha} 
      \ee
     We derive
     \be
     \lb{transE-SU3}
     E_\alpha  \ \stackrel \Omega  \to\   \frac 1{\sqrt{2}} (E_\alpha - i E_{-\beta}), &\qquad& E_\beta  \ \stackrel \Omega\to  \ \frac 1{\sqrt{2}} (E_\beta + i E_{-\alpha}), \nn
      E_{-\alpha}  \ \stackrel \Omega \to \ \frac 1{\sqrt{2}} (E_{-\alpha} + i E_{\beta}), &\qquad& E_{-\beta}  \ \stackrel \Omega\to  \ \frac 1{\sqrt{2}} (E_{-\beta} - i E_{\alpha}) \,.
       \ee

As was also the case for $I$, the complex structure $J$  splits into two blocks. The first block describes the action of $J$ in the subspace (4,5,3,8). By construction, it coincides with \p{block-J}.   The second block describes the action of $J$ in the subspace (1,2,6,7). Bearing in mind \p{act-I-korni} and \p{transE-SU3}, we derive
 \be
 \hat J E_\alpha =  -E_{-\beta}, \quad  \hat J E_\beta =   E_{-\alpha}, \quad  \hat J E_{-\alpha} =  -E_{\beta}, \quad
  \hat J E_{-\beta} = E_\alpha
   \ee
   or 
   \be
   \lb{act-J-korni}
   \hat J t_1 = -t_6, \quad     \hat J t_2 = t_7, \quad  \hat J t_6 = t_1,  \quad \hat J t_7 = -t_2 \,.
    \ee
    This  gives the matrix  that differs from \p{block-J} by the sign. In other words, $J = {\rm diag} (\EuScript{J}, -\EuScript{J})$. The third complex structure is  $K = {\rm diag} (\EuScript{K}, -\EuScript{K})$.
    
    \subsection{Higher unitary groups}
    
 Consider first $SU(4) \times U(1)$.  The positive roots of $SU(4)$ are schematically shown below: 
      \be
 \lb{roots-SU4}
 \left(  \begin{array}{cccc}  * & \alpha & \alpha+\beta & \alpha+\beta+\gamma \\
 *&*& \beta & \beta+\gamma \\
 *&*&*& \gamma \\
 *&*&*&* 
                             \end{array} \right)\,.
       \ee
  The action of the complex structure $I$ on the root vectors is defined in \p{act-I-korni}. We choose the basis of the Cartan subalgebra of $SU(4) \times U(1)$ as follows:
 \be
 \lb{CSA-SU4}
 t_{\rm out} = \frac 12 (\alpha+\beta+\gamma)^\vee \ =\  \frac 12 {\rm diag} (1,0,0,-1), &\qquad&  t_{\rm in} = \frac 12 \beta^\vee  = \frac 12 {\rm diag} (0,1,-1,0), \nn
 t_{15} = \frac 1{2\sqrt{2}}  {\rm diag} (1,-1,-1,1), &\qquad&  t_0 = \frac 1{2\sqrt{2}}  {\rm diag} (1,1,1,1)\,.
  \ee
  We then define\footnote{One could interchange $t_{15}$ and $t_0$  in \p{act-I-CSA-SU4}---it does not matter so much. Also the signs in \p{act-I-CSA-SU4} could be chosen differently, but we are trying to establish the universal recipe that will work for all groups.}
  \be
  \lb{act-I-CSA-SU4}
  \hat I t_{\rm out} = -t_{15}, \quad  \hat I t_{15} = t_{\rm out}, \quad  \hat I t_{\rm in} = -t_{0}, \quad \hat I t_0 = t_{\rm in}\,.
   \ee
   The matrix $I$ has a block-diagonal form. We distinguish four $4 \times 4$ blocks: {\it (I)} the outer block including the generators
   $ E_{\pm(\alpha + \beta +\gamma)},  t_{\rm out}$ and $t_{15}$; {\it (II)} the internal block    including the generators
   $ E_{\pm\beta}, t_{\rm in}$ and $t_0$; {\it (III)} the block including $E_{\pm \alpha}$ and $ E_{\pm(\beta + \gamma)}$ and {\it (IV)} the block including  $E_{\pm(\alpha+\beta)}$ and  $E_{\pm\gamma}$. All these blocks have the form \p{block-I}.
 
  In the analogy with the previous examples, we are trying to construct the second complex structure as $\tilde J = \Omega I \Omega^T$ where 
  $\Omega$ is the automorphism $t_a \to U_{\rm out}^\dagger t_a U_{\rm out}$ with 
   \be
   \lb{Uout}
   U_{\rm out} \ =\ \exp \left\{ \frac {i\pi}4 (E_{\alpha+\beta+\gamma} + E_{-\alpha-\beta-\gamma}) \right\}\,.
    \ee
    One immediately sees that this gives the structure \p{block-J} in the  ``outer" sector {\it (I)}. 
    
    The nonzero commutators involving
    $E_{\pm(\alpha+\beta+\gamma)}$ are 
    \be
    \lb{commE}
    [E_{\pm (\alpha+ \beta+\gamma)}, E_{\mp \alpha}] =  \ \mp E_{\pm (\beta+\gamma)}, \qquad 
    [E_{\pm (\alpha+ \beta+\gamma)}, E_{\mp (\beta+\gamma)}] =  \ \pm E_{\pm \alpha}, \nn
    \, [E_{\pm (\alpha+ \beta+\gamma)} , E_{\mp \gamma}] =  \ \pm E_{\pm (\alpha+\beta)}, \qquad 
     [E_{\pm (\alpha+ \beta+\gamma)}, E_{\mp (\alpha+ \beta)}] =  \ \mp E_{\pm \gamma}\,.
      \ee
    Then the action of $\hat {\tilde J}$ on the root vectors $E_{\pm \alpha},  E_{\pm (\beta+\gamma)}, E_{\pm (\alpha+ \beta)},
    E_{\pm \gamma}$ reads
    \be
    \lb{act-tilde-J}
 \hat {\tilde J} E_\alpha =  -E_{-\beta-\gamma}, \quad  \hat {\tilde J} E_{\beta+ \gamma} =  E_{-\alpha}, \quad  \hat {\tilde J} E_{-\alpha} =  -E_{\beta+\gamma}, \quad
  \hat {\tilde J} E_{-\beta-\gamma} = E_\alpha \,, \nn
  \hat {\tilde J} E_{\alpha+\beta} =  -E_{-\gamma}, \quad  \hat {\tilde J} E_{\gamma} =  E_{-\alpha-\beta}, \quad  \hat {\tilde J} E_{-\alpha-\beta} =  -E_{\gamma}, \quad
  \hat {\tilde J} E_{-\gamma} = E_{\alpha+\beta} \,. 
   \ee
   We see that $\hat {\tilde J}$ does not mix the sectors {\it (III)} and {\it (IV)} neither with the sectors {\it (I)}, {\it (II)}, nor with each other, and that the matrix $\tilde J$ in these sectors has the same form as 
   in \p{act-J-korni}, coinciding up to the sign with \p{block-J}.
   
   However, the internal block {\it (II)} of the matrix $\tilde J$ is not affected by the automorphism \p{Uout}. Indeed, its generators   $ E_{\pm\beta}, t_{\rm in}, t_0$ {\it commute} with $E_{\pm(\alpha+\beta+\gamma)}$, they belong to the {\sl centralizer} of $E_{\pm(\alpha+\beta+\gamma)}$   in $su(4)$. As a result, the internal block in $\tilde J$  coincides with that in $I$. Thus, the matrix $\tilde J$ does not anticommute with $I$ and is not a suitable choice for the second complex structure.
   
   Well, it is easy to understand what we should do next. We should apply to $\tilde J$ the ``internal" automorphism $t^a \to U_{\rm in}^\dagger t_a U_{\rm in}$ with 
    \be 
    \lb{Uin}
    U_{\rm in} \ =\ \exp \left\{ \frac {i\pi}4 (E_\beta + E_{-\beta}) \right\}\,.
     \ee
     This automorphism does not act on $E_{\pm(\alpha + \beta + \gamma)}$ and hence the outer block {\it (I)} is left unchanged.
     It acts on the internal block {\it (II)} in the same way as the automorphism \p{IJ-auto} for $SU(2)$, transforming $\EuScript{I}$ 
     to $\EuScript{J}$. But it also acts nontrivially on the root vectors $E_{\pm \alpha}, E_{\pm \gamma}, E_{\pm(\alpha+\beta)}, E_{\pm (\beta+\gamma)}$ with a potential danger that the "good" structure of the blocks {\it (III)} and {\it (IV)} would be spoiled.
     Fortunately, this does not happen.
    
   Indeed, the only nonzero commutators that one has to take into account in order to determine the action of $\Omega_{\rm in}$ on $E_{\pm \alpha}$ are $[E_\beta, E_\alpha] = -E_{\alpha+\beta}$ and $[E_{-\beta}, E_{-\alpha}] = E_{-\alpha-\beta}$. We derive
    \be
    U_{\rm in}^\dagger E_\alpha U_{\rm in} = \frac 1{\sqrt{2}} (E_\alpha + i E_{\alpha+\beta}), \qquad
    U_{\rm in}^\dagger E_{-\alpha} U_{\rm in} = \frac 1{\sqrt{2}} (E_{-\alpha} -  i E_{-\alpha-\beta}), 
     \ee
and similarly for $E_{\pm \gamma}, E_{\pm(\alpha+\beta)}, E_{\pm (\beta+\gamma)}$. In other words, the automorphism
$\Omega_{\rm in}$
mixes the positive root vectors from the ``outer layer" of  $su(4)$ with the positive root vectors and the negative root vectors    with the negative ones. However, the action \p{act-tilde-J} of $\hat {\tilde J}$ is the same for the doublets $(E_\alpha, E_{-\beta-\gamma})$ and  $(E_{\alpha+\beta}, E_{-\gamma})$ and is the same for the doublets $(E_\gamma, E_{-\alpha-\beta})$ and  $(E_{\beta+\gamma}, E_{-\alpha})$. Hence the blocks {\it (III)} and {\it (IV)} in the matrix $J = \Omega_{\rm in} \tilde J \Omega_{\rm in}^T$ have the same form as in $\tilde J$, coinciding with $-\EuScript{J}$. Thus, the matrix
 \be 
 J \ = \ (\Omega_{\rm in} \Omega_{\rm out}) \,I\, (\Omega_{\rm in} \Omega_{\rm out})^T
 \ee
 has the form  $\pm \EuScript{J}$ in all its four blocks. It anticommutes with $I$.
 
 The procedure for the higher unitary groups $SU(2l) \times U(1)$ and $SU(2l+1)$ is now clear. 
 \begin{itemize}
 \item Take e.g. $SU(7)$ with six simple roots $\alpha_j$. Take the highest root  $\theta_0 =  \sum_j^6 \alpha_j$. Consider the centralizer of  $E_{\pm \theta_0}$ in $su(7)$. It is $su(5)$ with the simple roots $\alpha_{j=2,3,4,5}$. Take the highest root there: $\theta_1 = \sum_{j=2}^5 \alpha_j$. The centralizer of $E_{\pm \theta_1}$ in $su(5)$ is $su(3)$ with the highest root
 $\theta_2 = \alpha_3 + \alpha_4$. This Russian doll construction terminates at this point, because the centralizer of $E_{\pm \theta_2}$ in $su(3)$ is $u(1)$, which is Abelian  with no further roots. Call the roots $\{\theta_0, \theta_1, \theta_2\}$ the {\it basic roots} of $\mathfrak g$   \cite{Spindel}. Endow each pair $(E_{\theta_k}, E_{-\theta_k})$ with the corresponding coroot $ t^{\rm CSA}_k =  \frac 12 \theta_k^\vee$. Choose three other orthonormal basic vectors $e^{\rm CSA}_k$ in the Cartan subalgebra in an arbitrary way. Define the action of $\hat I$ on the root vectors as in \p{act-I-korni} and supplement it with 
  \be
  \lb{act-I-CSA}
  \hat I t^{\rm CSA}_k =  e^{\rm CSA}_k\,, \qquad \hat I e^{\rm CSA}_k = -t^{\rm CSA}_k\,.
    \ee 
 Then $I$ has a block-diagonal form with the matrix $\EuScript{I}$ in each block.
 
\item Calculate $J^{(0)} = \Omega_0 I \Omega^T_0$ where $\Omega_0$ is the automorphism $t_a \to U^\dagger_0 t_a U_0$ with
 \be
 U_0 \ =\ \exp \left\{ \frac {i\pi}4 (E_{\theta_0} + E_{-\theta_0}) \right\}\,.
  \ee
This automorphism transforms $\EuScript{I}$ to $\EuScript{J}$ or to $-\EuScript{J}$ for the blocks of the ``outer layer". One of these blocks includes
$E_{\pm \theta_0}, t^{\rm CSA}_0$ and $e^{\rm CSA}_0$. For $SU(7)$, the outer layer also includes five blocks associated with
the doublets $\left(E_{\alpha_1}, E_{-\sum_{j=2}^6 \alpha_j}\right)$, \ldots ,    $\left(E_{\sum_{j=1}^5 \alpha_j}, E_{-\alpha_6}\right)$ and their
complex conjugates. 

Six other $4\times4$ blocks in $J^{(0)}$ still have the form $\EuScript{I}$.

\item We transform $J^{(0)}$ further going to $J^{(1)} = \Omega_1 J^{(0)} \Omega_1^T$ with
 \be
 U_1 \ =\ \exp \left\{ \frac {i\pi}4 (E_{\theta_1} + E_{-\theta_1}) \right\}\,.
  \ee 
  After that, the block $(E_{\pm \theta_1}, t^{\rm CSA}_1, e^{\rm CSA}_1)$ and three other blocks in the ``middle layer" of $su(7)$ [and the outer layer of $su(5) \subset su(7)$] are converted from $\EuScript{I}$ to  $\pm\EuScript{J}$.  The blocks in the outer layer do not change their form. 
  \item 
  Two remaining unconverted blocks in the internal $su(3)$---the block  involving the generators $(E_{\pm \theta_2}, t_2^{\rm CSA}, e_2^{\rm CSA})$ and the block $(E_{\pm \alpha_3}, E_{\pm \alpha_4})$---are converted by the automorphism involving
  \be
 U_2 \ =\ \exp \left\{ \frac {i\pi}4 (E_{\theta_2} + E_{-\theta_2}) \right\}\,.
  \ee 
 
  We obtain thus the matrix 
  \be
  J \ =\ (\Omega_2 \Omega_1 \Omega_0)\, I  \,  (\Omega_2 \Omega_1 \Omega_0)^T\,,
   \ee
   which anticommutes with $I$.
   
   \item The third member of the quaternion triple we were looking for is $K = IJ$ or else
   \be
  K \ =\ (\tilde \Omega_2 \tilde \Omega_1 \tilde \Omega_0)\, I  \,  (\tilde \Omega_2 \tilde \Omega_1 \tilde \Omega_0)^T  
  \ee
  where $\tilde \Omega_k$ is the automorphism $t_a \to \tilde U^\dagger_k t_a \tilde U_k$ with 
 \be
 \tilde U_k \ =\ \exp \left\{ \frac {\pi}4 (E_{\theta_k} - E_{-\theta_k}) \right\}\,.
  \ee
\end{itemize}

\subsection{General construction}

It is not so difficult to generalize the construction outlined in the previous section to an arbitrary Lie group. Let the group $G$ represent a product of a non-Abelian simple block $\tilde G$ and some number of $U(1)$ factors (a generalization on the case when $G$ involves several simple non-Abelian factors is straightforward).

We will describe this generalization below, but, to make the construction more clear, we will illustrate it in the nontrivial example of 
the algebra $B_3 \equiv spin(7)$. Before doing that, let us recall the salient features of this algebra and the corresponding group.

The spinor representation of this group is 8-dimensional, and its generators may be represented by the matrices $T_{jk} = 
i\gamma_j \gamma_k/2$, where $T_{jk}$ is the generator of the spinor rotations in the $(jk)$ plane. $\gamma_j$ are the Euclidean Dirac matrices satisfying the Clifford algebra $\gamma_j \gamma_k + \gamma_k \gamma_j = 2\delta_{jk} \mathbb{1}$.
The  commutation relations are
 \be
 [T_{jk}, T_{mn}] \ =\ -i(\delta_{jm} T_{kn} - \delta_{jn} T_{km} + \delta_{kn} T_{jm} - \delta_{km} T_{jn})\,.
  \ee
 The rank of $spin(7)$ is 3, the generators $T_{12}, T_{34}$ and $T_{56}$ constitute the basis of the CSA.  The simple roots are 
\be
 \alpha \ = \ (1,-1,0),  \qquad \beta\  = \ (0,1,-1), \qquad \gamma \ = \ (0,0,1)\,.
\ee 
Two of them are long and the third is short. The corresponding simple coroots are 
\be
 \alpha^\vee  \ = \ T_{12} - T_{34};  \qquad \beta^\vee  = \ T_{34} - T_{56}; \qquad \gamma^\vee \ = \ 2T_{56}\,.
\ee 
Two of them are short and the third is long.
 Besides the simple roots, the algebra includes two  other short and four  long roots:
 \be
 \beta + \gamma  = (0,1,0); \quad \alpha+\beta+\gamma  = (1,0,0); \quad 
 \alpha+\beta = (1,0,-1); \nn \quad \beta + 2\gamma  = (0,1,1); \quad \alpha + \beta + 2\gamma  = (1,0,1); \quad 
  \theta = \alpha + 2\beta + 2\gamma  = (1,1,0)\,.
 \ee
 The corresponding coroots are
 \be
  (\beta + \gamma)^\vee = 2T_{34}; \quad  (\alpha+\beta+\gamma)^\vee  = 2T_{12}; \quad
 (\alpha+\beta)^\vee  = T_{12} - T_{56}; \nn
  (\beta + 2\gamma)^\vee  =  T_{34} + T_{56};  \quad
  \quad (\alpha + \beta + 2\gamma)^\vee   =  T_{12} + T_{56}; \quad 
  (\alpha + 2\beta + 2\gamma)^\vee  =  T_{12} + T_{34}\,.
 \ee
 All the coroots satisfy \p{coroot-norm}.
  We choose the Chevalley normalization for the root vectors. Then, for example, $E_\gamma$ may be chosen as $E_\gamma = T_{57} - i T_{67}$ so that $[E_\gamma, E_{-\gamma}] = \gamma^\vee$.

    \begin{figure} [ht!]
      \bc
    \includegraphics[width=\textwidth]{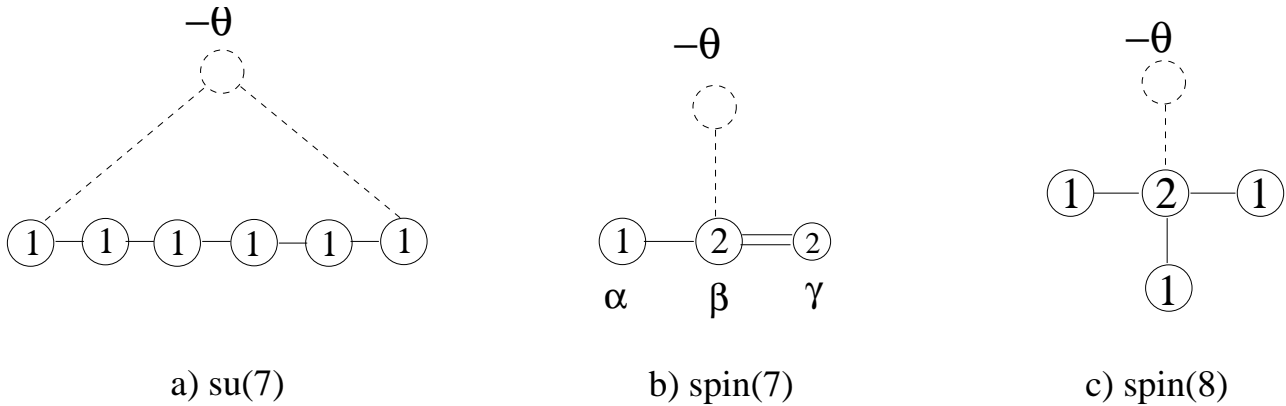}                  
     \ec
    \caption{Extended Dynkin diagrams for some algebras. $-\theta$ is the lowest root. A small circle stands for the short simple root in $spin(7)$. The numbers in the circles are the Dynkin labels---the factors with which the simple roots enter $\theta$. }        
 \label{Dynin}
    \end{figure}  
    
We proceed with our construction.

\begin{itemize}

\item Take the highest root $\theta$ in ${\mathfrak g}$. Consider the centralizer ${\mathfrak g}^{(1)}$ of 
$E_{\pm \theta}$ in $\mathfrak g$. 
An essential difference with the situation for the unitary groups is that this centralizer includes generically {\it several}
non-Abelian factors. The simplest algebra when it happens is $spin(7)$. Its extended Dynkin diagram including the simple   roots and the lowest root is drawn in Fig. \ref{Dynin}b.

The Dynkin diagram of the non-Abelian part of the centralizer of $E_{\pm \theta}$ in any Lie algebra $\mathfrak g$ is obtained from its extended Dynkin diagram by crossing out the circle $-\theta$ {\it and} also the circles for the simple roots with which  $-\theta$ is connected.
For $su(7)$ such a surgery gives the Dynkin diagram for $su(5)$. For $spin(7)$ we are left with two disconnected circles meaning that the centralizer is ${\mathfrak g}^{(1)} =  su(2) \oplus su(2) \oplus$ {\sl possible Abelian summands}, which are absent in this case. Such a centralizer involves two highest roots.

\item Take the highest root(s) in  ${\mathfrak g}^{(1)}$ and determine their centralizer ${\mathfrak g}^{(2)}$. 
If ${\mathfrak g}^{(2)}$ still includes non-Abelian summands, repeat the procedure. All the highest roots thus found constitute a set of the {\sl basic roots}.  To define the matrix $I$, we use \p{act-I-korni} and complement it as in \p{act-I-CSA}, where $t_k^{\rm CSA}$ are the basic coroots (with a factor $1/2$) and 
$e_k^{\rm CSA}$ are the remaining generators of CSA. The latter can be chosen arbitrary, one has only to take care that they are orthogonal to $t_k^{\rm CSA}$, to each other, and are normalized in the same way as $t_k^{\rm CSA}$.

Obviously, the number of $e_k^{\rm CSA}$ should coincide with the number of $t_k^{\rm CSA}$, and this imposes  a constraint on $G$. The group manifold $SU(2l+1)$ is HKT: all the centralizers ${\mathfrak g}^{(1)}, {\mathfrak g}^{(2)}, \ldots$ include only one simple non-Abelian summand with only one highest root, there are altogether $l$ basic roots, $l$ basic coroots, and there are exactly $l$ generators in the CSA, which are left. A similar counting works for $SU(2l) \times U(1)$. It also works for the group $Spin(7) \times [U(1)]^3$ (the dimension of $Spin(7)$ is 21, and we need to bring about three extra $U(1)$ factors to make the net dimension a multiple integer of 4). The set of the basic roots
in $spin(7)$ involves the highest root
$\theta_{spin(7)} \ =\ \alpha + 2\beta + 2\gamma$ and the highest roots $\alpha$ and $\gamma$ of the centralizer. Three corresponding coroots $t_k^{\rm CSA}$ are matched by three generators $e_k^{\rm CSA}$ of the three $U(1)$ factors.

But not any group manifold of dimension $4n$ is HKT. For example, $Spin(8)$ of dimension 28 is complex due to the Samelson theorem, but not HKT. Indeed, look at the extended Dynkin diagram of $spin(8)$ in Fig.\ref{Dynin}c. The  centralizer of the highest root is $su(2) \oplus su(2) \oplus su(2)$. This gives altogether $1+3 = 4$ basic roots and 4 corresponding coroots. Thereby, the CSA of $Spin(8)$ is exhausted and we cannot match $t_k^{\rm CSA}$ with $e_k^{\rm CSA}$. To get the latter, we have to endow 
$Spin(8)$ with four extra $U(1)$ factors. The manifold $Spin(8) \times [U(1)]^4$ is HKT. 

\item The highest root $\theta$ can be represented as a sum of two other roots in several different ways. For example, for $spin(7)$,
 \be
 \theta \ =\ (\alpha+ \beta) + (\beta + 2\gamma) \ =\ (\alpha+\beta+\gamma) + (\beta+\gamma) \ =\ 
 (\alpha+\beta+2\gamma) + \beta \,.
  \ee
  Let $\theta = \alpha^* + \beta^*$. In the Chevalley normalization for the root vectors, their commutators are given by 
   \p{E-norm}, \p{Bourbaki}. 
   In particular,
     \be
     \lb{commE-gen}
     [E_{\pm \theta}, E_{\mp \alpha^*}] \ =\ \pm E_{\pm \beta^*}, \qquad  [E_{\pm \theta}, E_{\mp \beta^*}] \ =\ \mp E_{\pm \alpha^*} \, .
       \ee
       Proceeding in the same way as for the unitary groups, we apply to $I$ the automorphism generated by the group element
       \be
       \lb{Uout-gen}
       U^{(0)} \ =\ \exp \left\{ \frac {i\pi}4 (E_\theta + E_{-\theta})\right\}\,.
        \ee
        Capitalizing on the fact that the commutators \p{commE-gen} have exactly the same form as in the unitary case [cf. \p{commE}],
        we deduce that the automorphism \p{Uout-gen} converts $\EuScript{I}$ to $\pm\EuScript{J}$ in the blocks associated with $E_{\pm \theta}$, $\theta^\vee$ and one more element of the CSA, and with the quartets $(E_{\pm \alpha^*}, E_{\pm \beta^*})$.
        \item The blocks associated with the centralizer ${\mathfrak g}^{(1)}$ stay unconverted at this stage. We convert some of them by the automorphism 
           \be
       \lb{Uin-gen}
       U^{(1)} \ =\ \prod_k \exp \left\{ \frac {i\pi}4 (E_{\theta_k} + E_{-\theta_k})\right\}\,,
        \ee
        where $\theta_k$ are the highest roots (sometimes, only one highest root) in ${\mathfrak g}^{(1)}$. If ${\mathfrak g}^{(2)}$ still includes non-Abelian summands, we repeat the procedure. If necessary, repeat it again...
        
        For $spin(7) \oplus u(1) \oplus u(1) \oplus u(1)$, the automorphism $U^{(0)}$ converts four outer blocks and the automorphism  $U^{(1)}$ converts two remaining blocks associated with $E_{\pm \alpha}$ and $E_{\pm \gamma}$. The complex structure $J$ thus obtained anticommutes with $I$.
        
\end{itemize}

\section{Homogeneous spaces}
\setcounter{equation}0

Not only group manifolds $G$ display the HKT structure. Some homogeneous spaces $G/H$ also do \cite{Joyce,OP}. It is rather easy to understand in our approach.

\vspace{1mm}

{\bf (i)} Consider $SU(4)$. As was explained above, the second complex structure $J$ is obtained in this case from $I$ by two consequent automorphisms \p{Uout} and \p{Uin}. The second automorphism was necessary to convert the internal block associated with the centralizer of $E_{\pm \theta}$, which for $su(4)$ is  $su(2) \oplus u(1)$. The  summand $su(2)$ includes $E_{\pm \beta}$ and $\beta^\vee$.
Suppose, however, that we are interested in the manifold $SU(4)/SU(2)$. Its tangent space involves only 12 generators of
$SU(4)$, not including $E_{\pm \beta}$ and $\beta^\vee$. But then the extra automorphism \p{Uin} is pointless. The outer automorphism \p{Uout} converting three blocks in the outer layer of $I$ is sufficient.

In fact, due to the presence of the $u(1)$ summand in the centralizer, we have several options at this point.
\begin{enumerate}
\item We can quotient $SU(4)$ over $SU(2)$  The tangent space of $SU(4)/SU(2)$ includes only two of
 the three generators of the CSA of $SU(4)$. They form together with $E_{\pm \theta}$  the outer $4 \times 4$ block. 
 \item We can quotient $SU(4)$ over the full centralizer 
 $SU(2) \times U(1)$ and multiply the result by $U(1)$.  The outer block includes in this case $E_{\pm \theta}$, the remaining generator of the CSA of $SU(4)$, and the generator of $U(1)$.
  \item We can quotient $SU(4)$ over $U(1)$ and multiply the result by  $[U(1)]^2$. In this case, the nontrivial internal block in the complex structures is left, and, to derive the complex structure $J$, we need to apply to $I$ both automorphisms    \p{Uout} and \p{Uin}.
   \item We can leave $SU(4)$ as it is and multiply it by $U(1)$. This gives the group manifold of the preceding section. 
\end{enumerate}
All these different manifolds are HKT.

\vspace{1mm}

{\bf (ii)} The simplest example is $SU(3)$. The centralizer of $E_{\pm(\alpha+\beta)}$ in $SU(3)$ is $U(1)$ generated by $t_8 \propto 
{\rm diag}(1,-2,1)$, and we can, besides  $SU(3)$, consider the coset $ \frac {SU(3)}{U(1)} \times U(1)$, which is also HKT. 

\vspace{1mm}
{\bf (iii)} For a higher unitary group like $SU(7)$, we have many possibilities. The centralizer of $E_{\pm \theta}$ in $SU(7)$ is $SU(5) \times U(1)$. This gives us the HKT manifolds $\frac {SU(7)}{SU(5)} $ and $\frac {SU(7)}{SU(5) \times U(1)} \times U(1)$. But our Russian doll construction has in this case two more levels. The root vectors $E_{\pm \theta_1}$, where $\theta_1$ is the highest root in $SU(5)$, have in $SU(5)$ a nontrivial centralizer $SU(3)\times U(1)$. Finally, the highest root of $SU(3)$ has the centralizer $U(1)$. One can factorize $SU(7)$ over any of these centralizers with or without the $U(1)$ factors and multiply the result over an appropriate number of $U(1)$ to obtain different HKT manifolds.

\vspace{1mm}

{\bf (iv)}
Our  final example are the homogeneous spaces associated with  $Spin(7)$. The centralizer of $E_{\pm \theta}$ in $Spin(7)$ is $G^{(1)} 
 = SU_\alpha (2) \times  SU_\gamma(2)$, and it is not simple. One can quotient $Spin(7)$ over $G^{(1)}$ and add the $U(1)$ factor to obtain the HKT manifold $\frac {Spin(7)}{SU(2) \times SU(2)} \times U(1)$, where the complex structures include four $4 \times 4$ blocks: the outer block associated with $\theta$ [it includes besides $E_{\pm \theta}$ the remaining generator of the CSA of $Spin(7)$ and the generator of $U(1)$] and three blocks associated with the pairs $(\alpha^*, \beta^*)$ satisfying $\alpha^* + \beta^* = \theta$. 
 
 Alternatively, one can quotient $Spin(7)$ over only one of the $SU(2)$ factors in $G^{(1)}$, for example, over $SU_\alpha(2)$.
 The relevant automorphism would in this case include the outer automorphism \p{Uout-gen} {\it and} the automorphism 
    \be
    \lb{U-gamma}
       U_\gamma \ =\ \exp \left\{ \frac {i\pi}4 (E_\gamma + E_{-\gamma})\right\}\,.
        \ee
        The automorphism \p{U-gamma} converts the block associated with the root $\gamma$, while the block associated with the root $\alpha$ is absent, and there is no need to convert it. We obtain the 20-dimensional HKT manifold $\frac {Spin(7)}{SU(2)} \times [U(1)]^2$.
        The block in the complex structures associated with $\theta$ includes besides $E_{\pm \theta}$ and $\theta^\vee$ the generator of one of $U(1)$, and the block associated with  $\beta$ includes besides $E_{\pm \beta}$ and $\beta^\vee$ the generator of another $U(1)$.
        
      It is clear now how to construct an arbitrary HKT homogeneous space.
      
      \begin{itemize}
      \item At  level 0, take any simple or semi-simple group $G$. As was shown in the preceding section, the group manifold $G \times [U(1)]^p$ including the appropriate\footnote{We hope that the reader understood from the considered examples what ``appropriate" means---the dimension of the CSA of   $G \times [U(1)]^p$ should be 2 times more than the number of the basic roots in the described Russian doll construction.} number of the $U(1)$ factors is HKT.
      \item Consider the centralizer $G^{(1)}$ of the set of the highest roots in  
      $G$. It includes one or several simple factors and $U(1)$ factors. Let ${\mathfrak S}_1$ be a set of all the products of these factors. 
      \item 
       An HKT homogeneous space of level 1 is the quotient of $G$ over any such product $s_1 \in  {\mathfrak S}_1$ multiplied by 
      an  appropriate number $p$ of $U(1)$ factors. $p$ should be such that the number $p+q$, where $q$ is the number of the CSA generators that are still left in the quotient, is twice as large as the number of the remaining basic roots.
      \item We may now pick up the highest roots in all the elements of ${\mathfrak S}_1$ that involve non-Abelian factors, consider their centralizers, define ${\mathfrak S}_2$ as the set of all the products of the simple factors and $U(1)$ factors of all these centralizers and repeat the procedure. We thus obtain the HKT homogeneous spaces of level 2.
       \item If possible, we may repeat the procedure again... 
      \end{itemize}
      The full list of the homogeneous spaces derived from all classical simple Lie groups $G$ is given in Table 2 of Ref. \cite{OP}.

\section{Deformations and supersymmetry}
\setcounter{equation}0

Our proof was constructed assuming the ``round" Killing metric \p{Killing} enjoying rich  isometry that corresponds to the left and right group multiplications.  But the Killing metric can be deformed so that the deformed metric is still HKT. 
Take the example of the Hopf manifold $SU(2) \times U(1) \equiv S^3 \times S^1$. Its metric can be written in the form
\be
\lb{Hopf-metr}
ds^2 \ =\ \frac {(dx_M)^2}{r^2}\,,
 \ee
 where $M = 1,2,3,4$, \ $r^2 = (x_M)^2$, and the points with the coordinates $x_M$ and $2x_M$ are identified. This metric is HKT. But it is known \cite{HKT} that  {\it any} metric obtained from \p{Hopf-metr} by a conformal transformation is still HKT. The metric ``squashed" in such a way would lose its isometries. Note  the complex structures 
 $(I_M{}^N,  J_M{}^N,  K_M{}^N)$ for all such metrics are the same, keeping the form \p{block-I}, \p{block-J}, \p{block-K}.
   
   The only known universal way to generate multidimensional HKT metrics is based on the {\sl harmonic superspace} technique. One starts by choosing certain harmonic prepotentials and then the metric can be found as a solution of a set of complicated harmonic equations. [The problem is thus essentially more intricate than in the K\"ahler case: any K\"ahler metric in a given chart can be represented as a double derivative of an arbitrary chosen K\"ahler potential, $h_{j \bar k} \ =\ \partial_j \partial_{\bar k} {\cal K} (z^m, \bar z^m)$.]  We will not describe this method here, referring
  the reader to the monography \cite{HSS}, to the paper \cite{DI}, where the problem of reconstructing the HKT metric in a framework of a certain ${\cal N} = 4$ supersymmetric sigma model in the harmonic superspace representation was solved, and to 
  \cite{FIS-nonlin} with a simple demonstration that the metric thus obtained is HKT, indeed.  
  
  The analysis of \cite{DI,FIS-nonlin} shows that, while for a 4-dimensional HKT manifolds, the allowed deformations include only one functional parameter, the conformal factor, one has much more freedom in higher dimensions. In dimension 8, one disposes of 6 such parameters. One obtains a large ``Obata family" of the HKT manifolds with coinciding Obata connections\footnote{The Obata connection \cite{Obata} is the torsionless connection with respect to which all three complex structures are covariantly constant. With such a choice, the metric need not be covariantly constant so that not only the directions of a tangent vector, but also its length may be changed under a parallel transport. We address the reader to Ref. \cite{FIS-nonlin} for detailed explanations.}  and coinciding complex structures.
  
  In recent \cite{DI-SU3}, the metric that seems to be a particular member of this family, which includes also the ``round" $SU(3)$ metric,  has been explicitly constructed. This metric 
  does not have the isometry $SU(3)\times SU(3)$, its isometry is only $SU(3)\times SU(2)$. The problem of  constructing
  the Killing metric \p{Killing} in the superspace approach rests by now unresolved.
  
  At the current state of knowledge, one cannot exclude a logical possibility that such a construction is not possible. This would mean that the supersymmetric technique suggested in \cite{DI} is not universal: it allows one to construct {\it some} HKT metrics, but not all of them. Personally, I think that this construction {\it is} universal, but the rigorous proof of this conjecture has not been given yet. One can mention, however, that such a universality proof exists for a similar  supersymmetric construction of  hyper-K\"ahler metrics: any such metric can be reproduced by choosing an appropriate prepotential \cite{L+4-HK}.
  We believe that a more attentive study of the supersymmetric constructions
 for $SU(3)$ and, eventually, reproducing in this way the Killing $SU(3)$ metric may help to construct a similar proof for all HKT manifolds.

 \section*{Acknowledgements}
 I am indebted to Evgeny Ivanov for many illuminating discussions and valuable comments on the manuscript and to Alexei Isaev and George Papadopoulos for useful correspondence and valuable remarks.


\begin{thebibliography}{96}


\bibitem{HKT}
 G.W.~Gibbons, G.~Papadopoulos and
 K.S.~Stelle,
{\it HKT and OKT geometries on soliton black hole moduli spaces},
Nucl. Phys.   {\bf B508} (1997) 623, {\tt arXiv:hep-th/9706207}; 

\bibitem{HKT-math}
G.~Grantcharov and Y.S.~Poon,
{\it Geometry of hyper-K\"ahler connections with torsion}, Commun. Math. Phys. {\bf 213}
(2000) 19, {\tt arXiv:math/9908015};
 M.~Verbitsky,
{\it Hyperk\"ahler manifolds with torsion, supersymmetry and Hodge theory},
Asian J. Math. {\bf 6} (2002) 679, {\tt arXiv:math/0112215}.

\bibitem{Spindel} P.~Spindel, A.~Sevrin, W.~Troost and A.~Van Proeyen, {\it Extended supersymmetric $\sigma$ models on group manifolds}, Nucl. Phys. {\bf B308} (1988) 662.

\bibitem{Coles} R.A.~Coles and G.~Papadopoulos, {\it The geometry of the one-dimensional supersymmetric non-linear sigma models}, Class. Quanum Grav. {\bf 7} (1990) 427.

\bibitem{FIS-nonlin}  S.A.~Fedoruk, E.A.~Ivanov and A.V.~Smilga, {\it Generic HKT geometries in the harmonic superspace approach},  J. Math. Phys. {\bf 59}, 083501 (2018), {\tt arXiv:  1802.09675 [hep-th]}.
 

\bibitem{DV} F.~Delduc and G.~Valent, {\it New geometry from heterotic supersymmetry}, Class. Quant. Grav. {\bf 10} (1993) 1201.

\bibitem{Joyce}
 D.~Joyce, {\it Compact hypercomplex and quaternionic manifolds}, J. Diff. Geom. {\bf 35} (1992) 743.
 
 \bibitem{OP} A. Opfermann and G. Papadopoulos, {\it Homogeneous HKT and QKT manifolds}, 
{\tt arXiv: math-ph/9807026}. 
 

 
\bibitem{NN}
 A.~Newlander and L.~Nirenberg, {\it Complex analytic coordinates in almost complex manifolds}, Ann. of Math. (2) {\bf 65} (1957) 391;
 L.~Nirenberg, {\it Lectures on linear partial differential equations}, Amer. Math. Soc., 1973.
 
 \bibitem{NN-ja} A.V.~Smilga, {\it Comments on the Newlander-Nirenberg theorem}, in: [Varna 2019, Proceedings, Lie Theory and its Application in Physics], arXiv:1902.08549 [math-ph].
 
 \bibitem{Bismut} N.E.~Mavromatos, {\it A note on the Atiyah-Singer index theorem for manifolds with totally antisymmetric $H$ torsion}, J. Phys.  {\bf A21} (1988) 2279;
 J.-M. Bismut, {\it A local index theorem for non K\"ahler manifolds}, Math. Ann. {\bf 284} (1989) 681.
 
 \bibitem{Hull} C.M.~Hull, {\it The geometry of supersymmetric quantum mechanics}, {\tt arXiv: hep-th/9910028}.
 
  \bibitem{Bourbaki} N. Bourbaki, {\it Lie groups and Lie algebras} [Springer, 2004], Chap. 8, \S 2, Proposition 7.
 
 \bibitem{Samelson} H.~Samelson, {\it A class of complex-analytic manifolds}, Portugal. Math. {\bf 12} (1953) 129.
 
 \bibitem{Hopf} H.~Hopf, {\it  Zur Topologie der komplexen Mannigfaltigkeiten}, in:[{\it Studies and Essays Presented to R. Courant on his 60th Birthday, January 8, 1948}, Interscience Publishers, Inc., New York], p.~167.
 
 \bibitem{Hooft} G.~'t Hooft, {\it  	
Computation of the quantum effects due to a four-dimensional pseudoparticle}, Phys. Rev. {\bf D14} (1976) 3432.


 

 
\bibitem{HSS}  A.S.~Galperin, E.A.~Ivanov, V.I.~Ogievetsky and E.S~Sokatchev, {\it Harmonic superspace}, Cambridge Univ. Press, 2001.

\bibitem{DI}  F.~Delduc and E.~Ivanov, {\it ${\cal N} = 4$ mechanics of general ({\bf 4}, {\bf 4}, {\bf 0}) multiplets}, Nucl. Phys. {\bf B855} (2012) 815, {\tt arXiv:1104.1429 [hep-th]}

\bibitem{Obata} M.~Obata, {\it Affine connections on manifolds with almost complex quaternion or Hermitian structures}, Japan J. Math. {\bf 26} (1956) 43.


\bibitem{DI-SU3} F.~Delduc and E.~Ivanov,  	
{\it ${\cal N} =4$
supersymmetric $d=1$ sigma models on group manifolds},
Nucl.Phys. {\bf B949} (2019) 114806, 
 {\tt arXiv:1907.09518 [hep-th]}.  

\bibitem{L+4-HK} A.S.~Galperin, E.A.~Ivanov, V.I.~Ogievetsky and E.S.~Sokatchev, {\it Gauge field geometry from complex and harmonic analyticities. II Hyper-K\"ahler case.}, Ann. Phys. {\bf 185} (1988) 22.


 

 
 
 
 


\end{thebibliography}
\end{document}